\documentclass{ws-ijgmmp}

\usepackage{amssymb,amsmath}
\usepackage{bm}%
\usepackage{stmaryrd}

\def\g{{\mathfrak g}}
\def\Z2{\mathbb{Z}_2^2}
\def\ph#1#2{(-1)^{\bm{#1}\cdot\bm{#2}}}
\def\f#1#2{f^{#1}_{\quad #2}}
\def\Dot#1#2{\bm{#1}\cdot\bm{#2}}

\newcommand{\NO}[1]{\vcentcolon\mathrel{#1}\vcentcolon}
\providecommand{\vcentcolon}{\mathrel{\mathop{:}}}

\begin{document}

\markboth{Aizawa \& Segar}
{Affine $\mathbb{Z}_2^2$-$osp(1|2)$ and Virasoro algebra}

%
\catchline{}{}{}{}{}
%

\title{Affine extensions of $\mathbb{Z}_2^2$-graded $osp(1|2)$ and Virasoro algebra
}

\author{N. AIZAWA}

\address{Department of Physics, Graduate School of Science, \\Osaka Metropolitan University \\
Nakamozu Campus, Sakai, Osaka 599-8531, Japan\\
\email{aizawa@omu.ac.jp} }

\author{J. SEGAR}

\address{Department of Physics, College of Engineering, \\Anna University,\\
	Guindy, Chennai 600025, India\\
	\email{segar@rkmvc.ac.in}}
\maketitle

\begin{history}
\received{(Day Month Year)}
\revised{(Day Month Year)}
\end{history}

\begin{abstract}
It is known that there are two inequivalent $\Z2$-graded $osp(1|2)$ Lie superalgebras. 
Their affine extensions are investigated and it is shown that one of them admits two central elements, one is non-graded and the other is $(1,1)$-graded. The affine $\Z2$-$osp(1|2)$ algebras are used by the Sugawara construction to study possible $\Z2$-graded extensions of the Virasoro algebra. We obtain a $\Z2$-graded Virasoro algebra with a non-trivially graded central element. 
Throughout the investigation, invariant bilinear forms on $\Z2$-graded superalgebras play a crucial role, so a theory of invariant bilinear forms is also developed. 
\end{abstract}


%
\section{Introduction}

The $\mathbb{Z}_2^n$-graded Lie superalgebra is a generalization of the Lie superalgebra in the sense that the $\mathbb{Z}_2$-grading of the Lie superalgebra is replaced by the more general abelian group  $\mathbb{Z}_2^n := \mathbb{Z}_2 \times \mathbb{Z}_2 \times \dots \times \mathbb{Z}_2$ ($n$ times) \cite{rw1,rw2}.  
This is a special case of the so-called colour Lie (super)algebra (a.k.a $\epsilon$-Lie algebra) \cite{Ree,sch}. 

In recent years, the $\mathbb{Z}_2^n$-graded Lie (super)algebras have attracted renewed interest in physics and mathematics for several reasons. 
They are able to generate symmetries and such a symmetry can be found in simple physical systems, such as  non-relativistic spinorial equations \cite{aktt1,aktt2,Ryan1}. Supersymmetry is also extended to the $\mathbb{Z}_2^n$-graded setting \cite{Bruce,akt1,brusigma} and $\mathbb{Z}_2^n$-graded versions of supersymmetric quantum mechanics have been extensively discussed  \cite{BruDup,AAD,DoiAi1,Topp,Topp2,Quesne, aiitotana2, BruSuperisation}. 
The $\mathbb{Z}_2^n$-extended supersymmetries, by definition, introduce new types of paraparticles which are different from those  of Green and Greenberg. 
%
We mention the works in \cite{parasim,paraexp} which prove the possibility of experimentally engineering paraparticles using trapped ions. 
This implies that the paraparticles could be real, and motivates us to further study various types of parastatistics.

The $\mathbb{Z}_2^n$-graded supersymmetries have a contact with higher supergeometry through the superfield formalism \cite{aido1,aizt,aiit,aiitotana}. 
The higher supergeometry is a $\mathbb{Z}_2^n$-graded extension of the supergeometry  and consider ``manifolds" whose local coordinates are an abelian $\mathbb{Z}_2^n$-graded Lie superalgebra (see \cite{Pz2nint,PonSch} for a review). 
The higher supergeometry is very non-trivial  due to the existence of the exotic bosonic coordinates and further investigations are needed to  understand it better.

At the level of algebra, the $\mathbb{Z}_2^n$-grading also causes non-triviality in structure and representations  \cite{tol,CART,StoVDJ,StoVDJ3,PhilNeliJoris,StoVDJ4,Meyer,Ryan2,ChenZhan}, in particular, we mention the $\Z2$-extension of the superdivision algebras \cite{FraZhan}.
 
The present work is motivated, among other things, by recent developments on the $\Z2$-graded integrable systems \cite{aiitotana,bruSG,niktt}. 
These works have elucidated the existence of a new class of integrable systems characterized by the $\Z2$-graded Lie superalgebras by extending the zero-curvature equations, Polyakov's soldering and superfield methods to $\Z2$-graded setting. 

Recall that the integrable sinh-Gordon equation can be represented as the zero-curvature equation for the affine $sl(2)$ Lie algebra \cite{BB}. 
This is extended to the $\Z2$-setting by introducing a $\Z2$-graded extension of the affine $sl(2)$ Lie algebra and a $\Z2$-graded version of the sinh-Gordon equation was obtained in \cite{niktt}.  

A natural next step is a supersymmetrization of the results of \cite{niktt}, as done in \cite{TopZhang} for the non-graded case. 
To do this, we need to introduce a $\Z2$-graded version of affine Lie superalgebras. 
The purpose of the present work is to consider this problem for the simplest case of $osp(1|2)$. 
Our starting point is a finite dimensional $\Z2$-graded $osp(1|2), $ but there are two of them, one is of eight dimension and the other is ten \cite{RY}.  
We shall see that, as in the non-graded case, invariant bilinear forms of the $\Z2$-graded $osp(1|2)$ play important roles and there exist an invariant  bilinear form with non-trivial $\Z2$-grading. Therefore, we organize this paper as follow:

In the next section, definitions of $\Z2$-graded Lie superalgebras and trace over a $\Z2$-graded matrix are given. Then  we introduce $\Z2$-graded bilinear forms using a $\Z2$-graded matrix. 
In \S \ref{SEC:osp}, the two inequivalent finite dimensional $\Z2$-graded $osp(1|2)$ are introduced and their bilinear forms are studied. A classification of the invariant bilinear forms is given. 
Affine extensions of the $\Z2$-graded $osp(1|2)$ of \S \ref{SEC:osp} are carried out in \S \ref{SEC:AffinExt}. It is shown that the affinization of the ten dimensional $\Z2$-graded $osp(1|2)$ admits two central extensions, while the eight dimensional one has the unique central element. 
Using the infinite dimensional algebras obtained in \S \ref{SEC:AffinExt}, we extend the Sugawara construction to the $\Z2$-setting to obtain a $\Z2$-graded extension of the Virasoro algebra.
Due to the two central elements of the affine $\Z2$-graded $osp(1|2)$, we obtain a $\Z2$-graded Virasoro algebra with two central charges. 
We summarize the results and mention some future perspectives in \S \ref{SEC:CR}.

\section{Bilinear Forms on $\Z2$-graded Lie Superalgebras}

\subsection{Basic definitions}

Here we give the definition of a $\Z2$-graded Lie superalgebra \cite{rw1,rw2}. 
Let $ \g $ be a vector space over $\mathbb{C}$ and $ \bm{a} = (a_1, a_2)$ an element of  $\Z2$. 
Suppose that $ \g $ is a direct sum of graded components:
\begin{equation}
	\g = \bigoplus_{\bm{a} \in \Z2} \g_{\bm{a}} = \g_{(0,0)} \oplus \g_{(1,1)} \oplus \g_{(0,1)} \oplus \g_{(1,0)}.
\end{equation}
For convenience, we introduce the \textit{standard} ordering of elements of $\Z2$ by
\begin{equation}
	(0,0), \ (1,1), \ (0,1), \ (1,0)
\end{equation}
We denote homogeneous elements of $ \g_{\bm{a}} $ as $ X_{\bm{a}}, Y_{\bm{a}},
Z_{\bm{a}}$ and call $ \bm{a}$ the \textit{degree} of $X_{\bm{a}}.$ 
If $\g$ admits a bilinear operation (the general Lie bracket), denoted by $ \llbracket \cdot, \cdot \rrbracket, $ 
satisfying the identities
\begin{align}
	& \llbracket X_{\bm{a}}, Y_{\bm{b}} \rrbracket \in \g_{\bm{a}+\bm{b}},
	\\[3pt]
	& \llbracket X_{\bm{a}}, Y_{\bm{b}} \rrbracket = -(-1)^{\bm{a}\cdot \bm{b}} \llbracket Y_{\bm{b}}, X_{\bm{a}} \rrbracket,
	\\[3pt]
	& (-1)^{\bm{a}\cdot\bm{c}} \llbracket X_{\bm{a}}, \llbracket Y_{\bm{b}}, Z_{\bm{c}} \rrbracket \rrbracket
	+ (-1)^{\bm{b}\cdot\bm{a}} \llbracket Y_{\bm{b}}, \llbracket Z_{\bm{c}}, X_{\bm{a}} \rrbracket \rrbracket
	+ (-1)^{\bm{c}\cdot\bm{b}} \llbracket Z_{\bm{c}}, \llbracket X_{\bm{a}}, Y_{\bm{b}} \rrbracket
	\rrbracket =0
	\label{gradedJ}
\end{align}
where
\begin{equation}
	\bm{a} + \bm{b} = (a_1+b_1, a_2+b_2) \in \Z2, \qquad \bm{a}\cdot \bm{b} = a_1 b_1 + a_2 b_2,
\end{equation}
then $\g$ is referred to as a $\Z2$-graded Lie superalgebra. 

We take $\g$ to be contained in its enveloping algebra, via the identification
\begin{equation}
	\llbracket X_{\bm{a}}, Y_{\bm{b}} \rrbracket =  X_{\bm{a}} Y_{\bm{b}} - (-1)^{\bm{a}\cdot \bm{b}}
	Y_{\bm{b}} X_{\bm{a}}, \label{gradedcom}
\end{equation}
where an expression such as $ X_{\bm{a}} Y_{\bm{b}}$ is understood to denote the associative product
on the enveloping algebra. 
In other words, by definition, in the enveloping algebra the general Lie bracket $ \llbracket \cdot, \cdot
\rrbracket $ for homogeneous elements coincides with either a commutator or anticommutator. 

This is a natural generalization of Lie superalgebra which is defined with a $\mathbb{Z}_2$-graded structure:
\begin{equation}
	\g = \g_{(0)} \oplus \g_{(1)}
\end{equation}
with 
\begin{equation}
	\bm{a} + \bm{b} = (a+b), \qquad \bm{a} \cdot \bm{b} = ab.
\end{equation}
It should be noted that  $ \g_{(0,0)} \oplus \g_{(0,1)} $ and $ \g_{(0,0)} \oplus \g_{(1,0)} $ are
subalgebras of $\g$ (with $\Z2$-grading).

From now on, we denote the basis of $\g$ by $ X^1, X^2, \dots, X^r $ with $r = \dim \g. $ 
With a slight abuse of notation, we write 
$ \deg{X^a} = \bm{a}$. Introducing the structure constant, one may write
\begin{equation}
	\llbracket X^a, X^b \rrbracket = \sum_c \f{ab}{c}\, X^c.
\end{equation} 
By definition of the general Lie bracket, we have
\begin{equation}
	\f{ab}{c} = -\ph{a}{b} \f{ba}{c}.
\end{equation}

 Define the adjoint action of $X^a$  on $ Y \in \g $ by
\begin{equation}
	ad X^a(Y) := \llbracket X^a, Y \rrbracket,
\end{equation}
then it is immediate to see from Eq.\eqref{gradedJ} that the adjoint action is an algebraic homomolphism
\begin{equation}
	ad \llbracket X^a, X^b \rrbracket = \llbracket ad X^a, ad X^b \rrbracket. \label{ad-hom}
\end{equation}
Thus, the adjoint action defines the adjoint representation of $\g$. 
The matrix elements of the adjoint representation is given by the structure constant as usual:
\begin{equation}
	(ad X^a)_{ij} = \f{aj}{i}.
\end{equation}
Arranging the basis of $\g$ according to the standard ordering of $\Z2$, matrix of the adjoint representation of a homogeneous element has non-vanishing entries only on the position indicated below:
\begin{equation} 
	\begin{bmatrix}
		A^{(0,0)} & A^{(1,1)} & A^{(0,1)} & A^{(1,0)} 
		\\
		B^{(1,1)} & B^{(0,0)} & B^{(1,0)} & B^{(0,1)}
		\\
		C^{(0,1)} & C^{(1,0)} & C^{(0,0)} & C^{(1,1)}
		\\
		D^{(1,0)} & D^{(0,1)} & D^{(1,1)} & D^{(0,0)}
	\end{bmatrix}. 
	\label{Z22matrix}
\end{equation}
The $\Z2$-graded version of the supertrace of the matrix Eq.\eqref{Z22matrix} is defined by \cite{rw2}
\begin{equation}
	Str X := Tr A^{(0,0)} + Tr B^{(0,0)} - Tr C^{(0,0)} - Tr D^{(0,0)}
\end{equation}
and has the properties
\begin{align}
	Str X &= 0 \quad \text{for} \quad  \deg(X) \neq (0,0),
	\label{str1} \\[3pt]
	Str(XY) &= \ph{a}{b}\, Str(YX), \quad X \in \g_{\bm{a}}, Y \in \g_{\bm{b}}
	\label{str2}
\end{align}

\subsection{Bilinear forms and $\Z2$-graded Casimir elements}

Let $ad X^a$ be the matrix of the adjoint representation of $\g$ (this abuse of notation will not cause any confusion). 
We introduce a $r \times r $ matrix $M$ of degree $\bm{m} \in \Z2 $ satisfying
\begin{equation}
	\llbracket M, ad X^a \rrbracket = 0, \quad \forall X^a \in \g \label{Mcommute}
\end{equation}

\begin{definition}
	A bilinear form $\eta$ on $\g$ is defined by
	\begin{equation}
		\eta^{ab} := \eta(X^a, X^b) :=  Str(ad X^a M ad X^b). \label{eta-def}
	\end{equation}
\end{definition}

\begin{proposition} \label{P:BLFs}
	The bilinear form $ \eta$ satisfies
	\begin{description}
		\renewcommand{\labelenumi}{(\roman{enumi})}
		\item[\ (i)] $ \eta^{ab} = 0 $ if $ \bm{a} + \bm{b} \neq \bm{m}$
		\item[\ (ii)] $ \eta^{ba} = (-1)^{\Dot{a}{b}+ \Dot{m}{m}} \eta^{ab}$
		\item[\ (iii)] $ \eta(\llbracket X^a, X^b\rrbracket, X^c) = (-1)^{\Dot{m}{b}}\, \eta(X^a,\llbracket X^b, X^c \rrbracket)$
	\end{description}
	(iii) is the adjoint invariance of $\eta.$ 
\end{proposition}
\begin{proof}
(i) follows from that $ Str(ad X^a M ad X^b) = 0 $ if $ \bm{a} + \bm{b} \neq \bm{m}.$ 

\noindent 
(ii) is proved by direct computation as follows:
\begin{align*}
	\eta^{ba} &=  Str(ad X^b M ad X^a) \stackrel{\eqref{str2}}{=} (-1)^{\bm{a} \cdot (\bm{b}+\bm{m})}  Str(ad X^a ad X^b M)
	\\
	&\stackrel{\eqref{Mcommute}}{=} (-1)^{\bm{a} \cdot (\bm{b}+\bm{m})} (-1)^{\Dot{b}{m}}  Str(adX^a M ad X^b) = (-1)^{ \Dot{a}{b}+\Dot{m}{m} } \eta^{ab}.
\end{align*}
The last equality follows from $ \bm{a}+\bm{b} = \bm{m}.$ 

\noindent
(iii) Using Eqs.\eqref{str2}, \eqref{Mcommute} and \eqref{ad-hom}, it is proved by direct computation. 
We set $ A := ad X^a, B := ad X^b, C := ad X^c$ for better readability.
\begin{align*}
	\eta(\llbracket X^a, X^b\rrbracket, X^c) &=  Str(\,ad\llbracket X^a, X^b \rrbracket M ad X^c\,)
	\\
	&=  Str( AB M C) - (-1)^{\Dot{a}{b}}  Str( B A M  C)
	\\
	&= (-1)^{\Dot{m}{b}}  Str(A M B C)
	- (-1)^{\Dot{a}{b}} (-1)^{\bm{b} \cdot (\bm{a}+\bm{m}+\bm{c})}  Str(A M C B)
	\\
	&= (-1)^{\Dot{m}{b}} \big( Str( A M B C) - (-1)^{\Dot{b}{c}}  Str(A M C B) \big)
	\\
	&= (-1)^{\Dot{m}{b}} \eta(X^a,\llbracket X^b, X^c \rrbracket).
\end{align*}
\end{proof}
In terms of the structure constants, Proposition \ref{P:BLFs} (iii) reads as follows:
\begin{equation}
		\sum_k \f{ab}{k} \eta^{kc} = \sum_k (-1)^{\Dot{m}{b}} \eta^{ak} \f{bc}{k}. \label{P13f}
\end{equation}
When the bilinear form $\eta = (\eta^{ab})$ is non-degenerate, we have the inverse $ \eta^{-1} = (\eta_{ab})$. In terms of $\eta^{-1}$, Eq.\eqref{P13f} yields
\begin{equation}
	\sum_k \eta_{ak} \f{kb}{c} = \sum_k \ph{m}{b} \f{bk}{a} \eta_{kc}. \label{P13finv}
\end{equation}

\noindent
\textbf{Remarks.}
\begin{enumerate}
	\item Proposition \ref{P:BLFs} (i) implies that $\eta $ is graded with $\deg(\eta) = \bm{m}$, since $\eta$ has the block form Eq.\eqref{Z22matrix} and its non-vanishing entries are only on the block with the suffix $\bm{m}$. 
	\item  One may take $ M = \mathbb{I}_r$ (identity matrix) as the $(0,0)$-graded matrix $M.$ Then $\eta$ is a natural $\Z2$-extension of the Killing form of Lie algebras. We use $g$ instead of $\eta $ for the $(0,0)$-graded bilinear form:
	\begin{equation}
		g^{ab} := K(X^a,X^b) :=  Str(ad X^a ad X^b) = \sum_{i,j} \ph{i}{j} \f{aj}{i} \f{bi}{j}. \label{Killing-def}
	\end{equation}
	From Proposition \ref{P:BLFs} we have
	\begin{equation}
		g^{ab} = \ph{a}{b} g^{ba},
	\end{equation}
	and
	\begin{equation}
		K(\llbracket X^a, X^b \rrbracket, X^c) = K(X^a, \llbracket X^b, X^c \rrbracket) 
		\quad \Leftrightarrow \quad \sum_d \f{ab}{d}\, g^{dc} = \sum_d g^{ad} \f{bc}{d}. 
	\end{equation}
	\item If $\det(M) = (1,1)$, then $ \eta $ is symmetric as $ \Dot{a}{b} = 0$ due to Proposition \ref{P:BLFs} (i). 
\end{enumerate}

The non-degenerate $\eta$ gives the Casimir elements of $\g:$
\begin{proposition} \label{P:Casimir}
	2nd Order $\Z2$-graded Casimir elements of $\g$ are given by
	\begin{align}
		& C_{\bm{m}} = \sum_{a,b} \eta_{ab} X^a X^b, \quad \bm{m} = \deg(\eta) = \deg(C_{\bm{m}}),
		\notag \\
		& \llbracket X^a, C_{\bm{m}} \rrbracket = 0, \quad \forall X^a \in \g
	\end{align}
\end{proposition}
\begin{proof} 
	Using the symmetries of $ \f{ab}{c}, \eta_{ab} $ and $ \f{ab}{c} = 0 $ if $ \bm{a}+\bm{b} \neq \bm{c}$, one may compute as follows:
	\begin{align*}
	\llbracket X^a, C_{\bm{m}} \rrbracket 
	&= \sum_{b,c} \eta_{bc} \big(\, \llbracket X^a, X^b \rrbracket X^c + \ph{a}{b} X^b \llbracket X^a, X^c \rrbracket \,  \big)
	\\
	&= \sum_{b,c,c} (\eta_{cd} \f{ac}{b} + \ph{a}{b} \f{ac}{d} \eta_{bc} ) X^b X^d
	\\
	&= \sum_{b,c,c} (-1)^{ 1+\Dot{m}{m} + (\bm{a}+\bm{d})\cdot (\bm{m}+\bm{d}) } (\eta_{dc} \f{ca}{b} - \ph{m}{a} \f{ac}{d} \eta_{cb}) X^b X^d = 0.
  \end{align*}
	The last equality is due to	Eq.\eqref{P13finv}.
\end{proof}

%
\section{$\Z2$-Graded $osp(1|2)$ Superalgebras} \label{SEC:osp}

As shown in \cite{RY}, there exist two inequivalent $\Z2$-graded extensions of the Lie superalgebra $ osp(1|2).$ One of them is eight dimensional and the other is ten dimensional. 
Each extension is simply denoted by $\g^8 $ and $ \g^{10}$, respectively. 

\subsection{$\g^8$ $:$ eight dimensional $\Z2$-$osp(1|2)$}

The basis of $\g^8$ is taken to be
\begin{equation}
	\begin{array}{ccccc}
		ad R \ & \; (0,0) \; &\; (1,1)\; &\; (0,1)\; &\; (1,0)
		\\[3pt]
		+2 &  L^+ & & & 
		\\
		+1 & & & a^+ & \tilde{a}^+
		\\
		0  & R & \tilde{R} & &
		\\
		-1 & & & a^- & \tilde{a}^-
		\\
		-2 & L^-
	\end{array}
\end{equation}
and they are subject to the non-vanishing relations:
\begin{alignat}{4}
	[R,L^{\pm}] &= \pm 2 L^{\pm}, & \qquad [L^{+},L^{-}] &= -R, & \qquad [R, a^{\pm}] &= \pm a^{\pm}, & \qquad  
	[ R, {\tilde{a}}^{\pm}] &= \pm {\tilde{a}}^{\pm}, 
	\notag \\[3pt]
	[L^{\pm}, a^{\mp}] &= \mp a^{\pm}, & 
	\{\tilde{R}, a^{\pm}\} &= \pm \tilde{a}^{\pm}, & \{\tilde{R}, \tilde{a}^{\pm}\} &= \mp a^{\pm}, &
	\{a^{+}, a^{-}\} &= 2R, 
	\notag \\
	[a^{\pm}, \tilde{a}^{\mp}] &= 2 \tilde{R}, & \{ \tilde{a}^{+}, \tilde{a}^{-} \} &= 2R, &  
	[L^{\pm}, \tilde{a}^{\mp}] &= \pm \tilde{a}^{\pm}, &
	\{a^{\pm}, a^{\pm}\} &= 4L^{\pm},
	\notag \\
	\{\tilde{a}^{\pm}, \tilde{a}^{\pm}\} &= -4L^{\pm}. 
\end{alignat}
We ordered the basis in the following way:
\begin{equation}
	(X^1, X^2, \dots, X^8) = (L^+, R, L^-,  \tilde{R},  a^+,a^-, \tilde{a}^+, \tilde{a}^-). \label{8Dorderedbasis}
\end{equation}

The invariant bilinear forms on $\g^8$ follow from the folowing observation:
\begin{proposition} \label{P:8DM0}
	The scalar matrix is the only one satisfying Eq.\eqref{Mcommute}. Namely, there is no matrix $M$ satisfying Eq.\eqref{Mcommute} with $ \deg(M) \neq 0.$ 
\end{proposition}	
\begin{proof}
	The assertion is verified by the straightforward computation. 
	For instance, suppose that $\deg(M) = (1,1)$. The condition $ [ad L^{\pm}, M] = [R, M] = 0$ puts some constraints on the entries of $M$. Solving all the constraints obtained from Eq.\eqref{Mcommute}, one may see that $M =0.$ Repeating this for the Matrix $M$ of the other degree, one may complete the proof.
\end{proof}

Proposition \ref{P:8DM0} means that the invariant bilinear form on $\g^8 $ is only the $\Z2$-version of the Killing form. 
The non-vanishing values of the $\Z2$-Killing form and its inverse are given by
\begin{align}
	g^{13}&= g^{31} =  -2,  \qquad\quad  
	g^{22} = g^{44} = 4,
	\notag \\
	g^{56} &= -g^{65} = g^{78} = -g^{87} = 8
\end{align}
and
\begin{align}
	g_{13} &= g_{31} = -\frac{1}{2}, \qquad \quad 
	g_{22} = g_{44} = \frac{1}{4}, 
	\notag \\
	g_{56} &= -g_{65} = g_{78} = - g_{87} = -\frac{1}{8}.
\end{align}
It follows that the 2nd order Casimir of $\g^8 $ is 
\begin{equation}
	C_2=\frac{1}{8}(2R^2 + 2\tilde{R}^2-4\{L^{+},L^{-} \} -[a^{+}, a^{-}]-[\tilde{a}^{+},\tilde{a}^{-}] )
\end{equation} 
which is identical to the one in \cite{FaFaJ}.

\subsection{$\g^{10}$ $:$ ten dimensional $\Z2$-$osp(1|2)$}

The basis of $\g^{10}$ is taken to be
\begin{equation}
\begin{array}{ccccc}
	ad R \ & \;(0,0)\; & \; (1,1)\; & \; (0,1)\; & \;(1,0)
	\\[3pt]
	+2 &  L^+ & \tilde{L}^+ & & 
	\\
	+1 & & & a^+ & \tilde{a}^+
	\\
	0  & R & \tilde{R} & &
	\\
	-1 & & & a^- & \tilde{a}^-
	\\
	-2 & L^- & \tilde{L}^- & & 
\end{array} 
\label{10Dbasis}
\end{equation}
and their non-vanishing relations are given by
\begin{alignat}{4}
	[R,L^{\pm}] &= \pm 2 L^{\pm}, &\qquad [R, \tilde{L}^{\pm}] &= \pm 2\tilde{L}^{\pm}, &\qquad  [R, a^{\pm}] &= \pm a^{\pm}, &\qquad  
	[ R, {\tilde{a}}^{\pm}] &= \pm {\tilde{a}}^{\pm}, 
	\notag \\[3pt]
	[\tilde{R},L^{\pm}] &= \pm 2\tilde{L}^{\pm}, & [\tilde{R}, \tilde{L}^{\pm}] &= \pm 2L^{\pm},& 
	\{\tilde{R}, a^{\pm}\} &= \tilde{a}^{\pm}, & \{\tilde{R}, \tilde{a}^{\pm}\} &= a^{\pm}, 
	\notag \\[3pt]
	[L^{+},L^{-}] &= -R, & [L^{\pm}, \tilde{L}^{\mp}] &= \mp \tilde{R}, & [\tilde{L}^{+}, \tilde{L}^{-}] &= -R, & 
	\notag \\[3pt]
	[L^{\pm}, \tilde{a}^{\mp}] &= \pm \tilde{a}^{\pm},& [L^{\pm}, a^{\mp}] &= \mp a^{\pm}, &
	\{ \tilde{L}^{\pm}, a^{\mp}\} &= -\tilde{a}^{\pm}, & \{\tilde{L}^{\pm}, \tilde{a}^{\mp}\} &= a^{\pm},
	\notag \\[3pt]
	\{a^{+}, a^{-}\} &= 2R, & [a^{\pm}, \tilde{a}^{\mp}] &= \pm 2 \tilde{R},& \{ \tilde{a}^{+}, \tilde{a}^{-} \} &= 2R, & 
	\notag \\[3pt]
	[a^{\pm}, \tilde{a}^{\pm}] &= \mp 4 \tilde{L}^{\pm}, & \{a^{\pm}, a^{\pm}\} &= 4L^{\pm},& 
	\{\tilde{a}^{\pm}, \tilde{a}^{\pm}\} &= -4L^{\pm}.
\end{alignat}
The ordering of the basis is taken to be
\begin{equation}
	(X^1, X^2, \dots, X^{10}) = (L^+, R, L^-, \tilde{L}^+, \tilde{R}, \tilde{L}^-, a^+,a^-, \tilde{a}^+, \tilde{a}^-). \label{Ordered10D}
\end{equation}

\begin{proposition} \label{P:10DMlist}
	The graded matrix $M$ satisfying Eq.\eqref{Mcommute} is unique up to an overall constant:
\begin{itemize}
	\item $\deg(M) = (0,0) \quad \Rightarrow \quad M = \mathbb{I}_{10}$
	\item $ \deg(M) = (0,1), (1,0) \quad \Rightarrow \quad M = 0$
	\item $ \deg(M) = (1,1) \quad \Rightarrow \quad \displaystyle 
	M = 
	\begin{bmatrix}
		0 & \mathbb{I}_3 & 0 & 0
		\\
		\mathbb{I}_3 & 0 & 0 & 0 
		\\
		0 & 0 & 0 & \sigma_3
		\\
		0 & 0 & \sigma_3 & 0
	\end{bmatrix}
	$ 
\end{itemize}
where $ \mathbb{I}_n, \sigma_3 $ is the $ n \times n $ unit matrix and the Pauli matrix, respectively. 
\end{proposition}
\begin{proof}
 The proof is similar to Proposition \ref{P:8DM0}, so we omit it.  	
\end{proof}

From Proposition \ref{P:10DMlist}, it can be seen that, contrary to $\g^8$, there exists  $(1,1)$-graded invariant bilinear form $\eta$ in addition to the $\Z2$-Killing form $g$. Their non-vanishing values are listed below:
\begin{align}
	g^{13} &= g^{31}  = g^{46} = g^{64}  = -6, \qquad
	g^{22} = g^{55} = 12,
	\notag \\
	g^{78} &= -g^{87} = g^{9\,10} = -g^{10\,9} = 24
\end{align}
and 
\begin{align}
	\eta^{16} &= \eta^{61} = \eta^{34} = \eta^{43} = -6, \qquad 
	\eta^{25} = \eta^{52} = 12,
	\notag \\
	\eta^{7\, 10} &= \eta^{10\, 7} = \eta^{89} = \eta^{98} = -24.
\end{align}
We also list the non-vanishing values of $g^{-1}$ and $\eta^{-1}:$
\begin{align}
	g_{13} &= g_{31}  = g_{46} = g_{64} = -\frac{1}{6}, \qquad
	g_{22} = g_{55} = \frac{1}{12},
	\notag \\
	g_{78} &= -g_{87} = g_{9\,10} = -g_{10\,9} = -\frac{1}{24}.
\end{align} 
and
\begin{align}
	\eta_{16} &= \eta_{61} = \eta_{34} = \eta_{43} = -\frac{1}{6}, \qquad 
	\eta_{25} = \eta_{52} = \frac{1}{12},
	\notag \\
	\eta_{7\, 10} &= \eta_{10\, 7} = \eta_{89} = \eta_{98} = -\frac{1}{24}.
\end{align}
It follows that $\g^{10}$ has $(0,0)$ and $(1,1)$-graded Casimirs \cite{FaFaJ}:
	\begin{align}
	C_{00} &=\frac{1}{24}(2R^2 + 2\tilde{R}^2-4 \{L^{+}, L^{-} \} -4\{\tilde{L}_{+}, \tilde{L}_{-} \}-[a^{+}, a^{-}] -[\tilde{a}^{+}, \tilde{a}^{-}] ),
	\notag \\
	C_{11} &=\frac{1}{24}(2 \{R, \tilde{R} \}-4 \{L^{+}, \tilde{L}^{-} \} -4\{\tilde{L}^{+}, L^{-} \}- \{a^{+}, \tilde{a}^{-} \}- \{\tilde{a}^{+}, a^{-} \}). 
	\label{11Cas}
\end{align} 

Irreducible representations of $\g^{10}$ are investigated in detail in \cite{NAKA}. 

%
\section{Affine Extensions of $\g^8$ and $ g^{10}$} \label{SEC:AffinExt}

\subsection{Loop algebra and central extension}

Affine $\Z2$-graded Lie superalgebras are defined in a similar way to the ordinary ones. 
First, we consider loop extension of a finite dimensional $\Z2$-graded Lie superalgebra $\g$ by
\begin{equation}
	L(\g) = \g \otimes \mathbb{C}[\lambda,\lambda^{-1}] 
\end{equation} 
where $ \lambda $ is a non-graded indeterminate. 
Denoting the basis of $ \g$ by $ X^a$, a basis of $ L(\g) $ is taken to be $ X^a_m := \lambda^m \otimes X^a, \ m \in \mathbb{Z}, \ \deg(X^a_m) = \bm{a} $ which satisfies
\begin{equation}
	\llbracket X^a_m, X^b_n \rrbracket = \f{ab}{c}\, X^c_{m+n}. 
\end{equation} 

We then consider $\Z2$-graded central extensions of $L(\g)$ defined by 
\begin{equation}
	\llbracket X^a_m, X^b_n \rrbracket = \f{ab}{c}\, X^c_{m+n} + \omega(X_m^a, X_n^b)
\end{equation}
where 
\begin{equation}
	\llbracket X^c_{k},\, \omega(X_m^a, X_n^b) \rrbracket = 0, \quad \forall X^c_k \in L(\g)
\end{equation}
It follows from this definition that 
\begin{enumerate}
	\item $ \deg \omega(X_m^a, X_n^b) = \bm{a} + \bm{b}$
	\item $ \omega $ is bilinear
	\item $ \omega(X_m^a, X_n^b) = -\ph{a}{b} \omega(X_n^b, X_m^a)$
	\item by the Jacobi identity Eq.\eqref{gradedJ}
	 \begin{align}
	 	\ph{a}{c} \omega(X_m^a,\llbracket X_n^b, X_k^c \rrbracket ) &+ \ph{b}{a} \omega(X_n^b,\llbracket X_k^c, X_m^a \rrbracket ) 
	 	\notag \\
	 	&+ \ph{c}{b} \omega(X_k^c,\llbracket X_m^a, X_n^b \rrbracket ) = 0. \label{2Ccond}
	 \end{align}
\end{enumerate}

A cohomological description of the $\Z2$-graded central extensions has been studied in \cite{SchZha}. 
Let $ I := \{  \ \omega(X_m^a, X_n^b) \ | \ a, b = 1, 2, \dots, \dim \g, \ m, n \in \mathbb{Z}  \ \} $ be the set of all central elements. 
It is then shown that there is a bijection between the equivalence class of central extensions of $L(\g)$ by $I$ and $\Z2$-extension of the 2nd cohomology group $ H^2(L(\g),I). $ 
The subset $ Z^2(L(\g),I) \subseteq I $ satisfying Eq.\eqref{2Ccond} is the 2-cocycle. 
The 2-coboundary (trivial extensions) is given by $ B^2(L(\g),I) = \{ \  \omega(X_m^a, X_n^b) = \omega(\llbracket X_m^a, X_n^b \rrbracket )  \ \} $ and $ H^2(L(\g),I) = Z^2(L(\g),I)/B^2(L(\g),I). $ 
Following \cite{SchZha}, we define the action of $ X \in L(\g) $ on $ \omega(A,B)$ by
\begin{equation}
	(X\omega)(A,B) := -(-1)^{ \bm{w}\cdot(\bm{a}+\bm{b}) } \omega(\llbracket X, A \rrbracket, B ) - \ph{w}{b} \omega(A,\llbracket X, B \rrbracket)
\end{equation}
where
\begin{equation}
	\bm{w} := \deg \omega = \bm{a}+ \bm{b}, \qquad \bm{a} := \deg A, \qquad \bm{b} := \deg B.
\end{equation}
Then, one may see that $ I$ is a $ L(\g)$-module under this action. 

To determine possible central extensions of $L(\g^8) $ and $ L(\g^{10})$, we use Proposition 2.1 of \cite{SchZha} which has been proved for the general colour Lie algebras. 
The proposition is rephrased in the present setting as follows:
\begin{proposition} \label{P:SZP21}
		Let $L'$ be a graded subalgebra of $L(\g)$. If the $L'$-module $Z^2(L(\g),I)$ is semi-simple, then for any $ g \in Z^2(L(\g),I)$, there exists an $L'$-invariant $ g' \in Z^2(L(\g),I)$ which is cohomologous to $g$.
\end{proposition}

One may also introduce $\Z2$-graded derivation $ d_{\bm{m}} $ of degree $\bm{m}$ by 
\begin{align}
	d_{\bm{m}} (\llbracket X^a_k, X^b_m \rrbracket) &= 
	\llbracket d_{\bm{m}}(X^a_k), X^b_m \rrbracket + \ph{m}{a} \llbracket X^a_k, d_{\bm{m}}(X^b_m) \rrbracket,
	\notag \\
	\llbracket d_{\bm{m}}, \omega \rrbracket &= 0.
\end{align}

Affine $\Z2$-graded Lie superalgebra $\widehat{\g}$ is defined as the centrally extended $ L(\g) $ and the $\Z2$-graded derivations.  

%
\subsection{$\widehat{\g}^8$ $:$ affine extension of $\g^8$}

\begin{theorem}
	All possible central extensions of $L(\g^8) $ are given by
	\begin{equation}
		\omega(X^a_m, X^b_n) = g^{ab} m \delta_{m+n,0}\, \omega \label{8Dcentext}
	\end{equation}
	where $\omega $ is the unique central element of degree $(0,0)$.
\end{theorem}
More explicitly, the non-vanishing relations with the central element are given as follows:
\begin{alignat}{2}
	[R_m,L^{\pm}_n] &= \pm 2 L^{\pm}_{m+n}, & \qquad [L^{+}_m,L^{-}_n] &= -R_{m+n}-2m \delta_{m+n,0}\, \omega, 
	\notag \\ 
	[R_m, R_n] &= 4m \delta_{m+n,0}\,\omega, &
	[R_m, a^{\pm}_n] &= \pm a^{\pm}_{m+n}, 
	\notag \\  
	[ R_m, {\tilde{a}}^{\pm}_n] &= \pm {\tilde{a}}^{\pm}_{m+n}, & 
	[L^{\pm}_m, a^{\mp}_n] &= \mp a^{\pm}_{m+n}, 
	\notag \\ 
	[\tilde{R}_m \tilde{R}_n] &= 4m \delta_{m+n,0}\,\omega, &
	\{\tilde{R}_m, a^{\pm}_n\} &= \pm \tilde{a}^{\pm}_{m+n}, 
	\notag \\
	\{\tilde{R}_m, \tilde{a}^{\pm}_n\} &= \mp a^{\pm}_{m+n},  &
	\{a^{+}_m, a^{-}_n\} &= 2R_{m+n} +8m \delta_{m+n,0}\,\omega, 
	\notag \\
	\{ \tilde{a}^{+}_m, \tilde{a}^{-}_n \} &= 2R_{m+n} + 8m \delta_{m+n,0}\,\omega, &
	[a^{\pm}_m, \tilde{a}^{\mp}_n] &= 2 \tilde{R}_{m+n},
	\notag \\
	[L^{\pm}_m, \tilde{a}^{\mp}_n] &= \pm \tilde{a}^{\pm}_{m+n}, &
	\{a^{\pm}_m, a^{\pm}_n\} &= 4L^{\pm}_{m+n},
	\notag \\
	\{\tilde{a}^{\pm}_m, \tilde{a}^{\pm}_n\} &= -4L^{\pm}_{m+n}.  
	\label{8DrelExplicit}
\end{alignat}
\begin{proof}
	Take $L' = \{  \  R_0 \ \}$ as the one-dimensional subalgebra of $L(\g^8)$. 
The action of $R_0$ on $\omega(X^a_m, X^b_n)$ is 
\begin{equation}
	(R_0 \omega)(X^a_m, X^b_n) = -\omega([R_0,X^a_m], X^b_n) -\omega(X^a_m, [R_0,X^b_n]) =-(x^a + x^b) \omega(X^a_m, X^b_n)
\end{equation}
where $ x^a$ is the eigenvalue of $ ad R_0 $. 
This means that, for a fixed $(a,b,m,n)$,  $ \omega(X^a_m, X^b_n) $ is one-dimensional module of $L'$. Thus, $ Z^2(L(\g),I)$ is semi-simple. 
By Proposition \ref{P:SZP21}, the non-vanishing central extension is possible only for $ x^a+x^b = 0.$ Namely, possible non-vanishing central elements are
\begin{equation}
	\begin{array}{ccc}
		\omega(R_m, R_n), & \quad \omega(\tilde{R}_m, \tilde{R}_n), & \quad \omega(L^+_m, L^-_n), 
		\\
		\omega(a^+_m, a^-_n), & \quad \omega(\tilde{a}^+_m, \tilde{a}^-_n), & \quad \omega(R_m, \tilde{R}_n),
		\\
		\omega(a^+_m, \tilde{a}^-_n), & \quad \omega(\tilde{a}^+_m, a^-_n).
	\end{array}	
\end{equation}
Note that $ L(\g^8) $ has two loop $osp(1|2)$ subalgebras, $ L(\mathfrak{h}_1) :=\{ \ R_m, \; L^{\pm}_m, \; a^{\pm}_m  \ \} $ and $ L(\mathfrak{h}_2) := \{ \ R_m, \; L^{\pm}_m, \; \tilde{a}^{\pm}_m  \ \}. $ 
It is known that the loop $osp(1|2)$ has the unique central extension. 
Since $L(\mathfrak{h}_1)$ and $ L(\mathfrak{h}_2)$ share the bosonic sector, the central extension of them must be common:
\begin{align}
		\omega(R_m,R_n) &= 4m \delta_{m+n,0}\, \omega, \quad \omega(L_m^+,L_n^-) = -2m\delta_{m+n,0}\, \omega,
		\notag \\
		\omega(a^+_m,a^-_n) &= \omega(\tilde{a}^+_m,\tilde{a}^-_n) = 8m \delta_{m+n,0}\, \omega.
\end{align}
 	
 	The  2-cocycle condition for $ \tilde{R}_m, L^{\pm}_m $ reads
\begin{equation}
 		\omega(\tilde{R}_m, [L^+_n,L^-_k]) + \omega(L^+_n, [L^-_k, \tilde{R}_m]) + \omega(L^-_k, [\tilde{R}_m, L_n^+]) = 0.
\end{equation}
It follows immediately that $ \omega(R_m, \tilde{R}_n) = 0. $ 
The 2-cocycle condition for $ \tilde{R}_m, a^+_n, \tilde{a}^-_k$ yields
\begin{equation}
 		2\omega(\tilde{R}_m,\tilde{R}_{n+k}) + \omega(a^+_m,a^-_{k+m}) + \omega(\tilde{a}^-_k,\tilde{a}^+_{m+n}) = 0.
\end{equation}
Set $k =0 $ then we have
\begin{equation}
 		2\omega(\tilde{R}_m,\tilde{R}_{n}) + \omega(a^+_m,a^-_{m}) + \omega(\tilde{a}^-_0,\tilde{a}^+_{m+n}) = 0.	
\end{equation}
It follows that $ \omega(\tilde{R}_m,\tilde{R}_{n}) = 4m \delta_{m+n,0}\,\omega. $

Finally, we show that $  \omega(a^+_m, \tilde{a}^-_n)$ and $  \omega(\tilde{a}^+_m, a^-_n) $ are trivial. 
The 2-cocycle condition for $R_n, a^{\pm}_m$ and $\tilde{a}^{\mp}_k$ yields
\begin{equation}
	\omega(a_m^{\pm}, \tilde{a}_{n+k}^{\mp}) - \omega(a^{\pm}_{m+n}, \tilde{a}^{\mp}_k) = 0.
	\label{2c-1}
\end{equation}
The 2-cocycle condition for  $\tilde{R}_n, a^{\pm}_m$ and $\tilde{a}^{\mp}_k$
\begin{equation}
	\omega(a^{\pm}_m, \tilde{a}^{\mp}_{n+k}) + \omega(\tilde{a}^{\pm}_{m+n}, a^{\mp}_k) = 0
	\label{2c-2}
\end{equation}
and for $ L^{\mp}_n, a^{\pm}_m $ and $ \tilde{a}^{\pm}_k$
\begin{equation}
	\omega(a^{\pm}_m, \tilde{a}^{\mp}_{n+k}) + \omega(\tilde{a}^{\pm}_k, a^{\mp}_{m+n}) = 0.
	\label{2c-3}
\end{equation}
Set $n=0$ in Eq.\eqref{2c-1} and Eq.\eqref{2c-2}, then summing them up we obtain
\begin{equation}
	\omega(a^{\pm}_m, \tilde{a}^{\mp}_k) + \omega(\tilde{a}^{\pm}_m,a^{\mp}_k) = 0. 
	\label{2c-4}
\end{equation}
Using these relations we see that $  \omega(a^+_m, \tilde{a}^-_n)$ and $  \omega(a^-_m,\tilde{a}^+_n) $ are equal:
\begin{equation}
	\omega(a^+_m,\tilde{a}^-_k) \stackrel{\eqref{2c-3}}{=} - \omega(\tilde{a}^+_k, a^-_m) = \omega(a^-_m,\tilde{a}^+_k).
\end{equation}
Furthermore, we see the followings
\begin{equation}
	\omega(a^{\pm}_m, \tilde{a}^{\mp}_{n}) \stackrel{\eqref{2c-2}}{=} -\omega(\tilde{a}_{m+n}^{\pm},a^{\mp}_0) \stackrel{\eqref{2c-4}}{=} \omega(a^{\pm}_{m+n},\tilde{a}^{\mp}_0)
\end{equation}
and
\begin{equation}
	\omega(a^{\pm}_m, \tilde{a}^{\mp}_{n}) \stackrel{\eqref{2c-3}}{=} -\omega(\tilde{a}^{\pm}_0,a^{\mp}_{m+n}) \stackrel{\eqref{2c-4}}{=} \omega(a^{\pm}_0, \tilde{a}^{\mp}_{m+n}). 
\end{equation}
These mean that $ \omega(a^{\pm}_m,\tilde{a}^{\mp}_n)$  is a function of $ m+n$ so that one may set
\begin{equation}
	\omega(a^{\pm}_m,\tilde{a}^{\mp}_n) = 2f(m+n)\, \omega(a,\tilde{a}), \quad
	\omega(a,\tilde{a}) = -\omega(\tilde{a},a).
\end{equation}
It follows that
\begin{equation}
	[a^{\pm}_m, \tilde{a}^{\mp}_n] = 2( \tilde{R}_{m+m} + f(m+n)\, \omega(a,\tilde{a})).
\end{equation}
Therefore, $ \omega(a^{\pm}_m,\tilde{a}^{\mp}_n) $ is absorbed into the redefinition of $\tilde{R}_m.$ 
This completes the proof.  	
\end{proof}	

Now we consider $\Z2$-graded derivations $d_{\bm{a}}$ defined by 
\begin{equation}
	\llbracket d_{\bm{a}}, X_n^b \rrbracket = n Y^c_n, \quad [d_{\bm{a}}, \omega ]  =0 
\end{equation}
where $ Y_n^c $ is a homogeneous element of degree $ \bm{c} = \bm{a} + \bm{b} $. 
This definition means that $d_{00}$ is the element which  resolve the degeneracy of $ ad R_0.$ 

\begin{proposition} \label{Derivation8}
	$d_{00}$ defined by $ [d_{00}, X^a_n] = n X^a_n $ is the derivation. However, $ d_{\bm{a}} $ does not exits for $ \bm{a} = 01, 10, 11.$ 
\end{proposition}
\begin{proof}
$d_{\bm{a}}$ is an element of $\Z2$-graded Lie superalgebra, so it must satisfy the $\Z2$-graded Jacobi identity. It is immediate to verify that $d_{00}$ satisfies the $\Z2$-graded Jacobi identity. 
The non-existence of $ d_{01} $ is seen from 
\begin{equation}
	\{ d_{01}, [R_0, a^+_m] \} = \{  d_{01}, a^+_m\}.
\end{equation}	
The lhs is rewritten by the Jacobi identity as 
\begin{equation}
	\mathrm{lhs} = [R_0, \{d_{01}, a^+_m\}].
\end{equation}
Therefore, $ \{  d_{01}, a^+_m\} $ is a $(0,0)$-graded eigenvector of $ ad R_0$ with the eigenvalue $+1.$ However, there is no such element in the algebra. 
The non-existence of $ d_{01}$ is shown similarly. Finally, from $ [d_{11}, [R_0, L^+_m]] $ we see that $ [d_{11}, L_m^+]$ is a $(1,1)$-graded eigenvector of $ ad R_0 $ with the eigenvalue $+2$. However, there is no such element.
\end{proof}

Gathering the results obtained so far, we defined the affine extension of $\g^8 $ by
\begin{equation}
	\widehat{\g}^8 = L(\g^8) \oplus \mathbb{C} \omega \oplus \mathbb{C} d_{00}.
\end{equation}

%
\subsection{$\widehat{\g}^{10}$ $:$ affine extension of $\g^{10}$}

\begin{theorem}
	$L(\g^{10})$ has two $\Z2$-graded central extensions. One is $(0,0)$-graded and the other is $(1,1)$-graded:
	\begin{equation}
		\omega(X^a_m, X^b_n) = (g^{ab} \, \omega_{00} + (-1)^{\bm{a}\cdot (1,1)} \eta^{ab} \omega_{11}) m \delta_{m+n,0}.
	\end{equation}
\end{theorem}
It follows that the non-vanishing relations with the central elements are given by
\begin{alignat}{2}
	[R_m, R_n] &= 12m\,\omega_{00}\,\delta_{m+n,0}, & \quad [R_m,\tilde{R}_n] &= 12m\,\omega_{11}\delta_{m+n,0},
	\notag \\[3pt]
	[\tilde{R}_m, \tilde{R}_n] &= 12m\,\omega_{00}\,\delta_{m+n,0}, & 
	[R_m,L^{\pm}_n] &= \pm 2 L^{\pm}_{m+n}, 
	\notag \\[3pt]
	[R_m, \tilde{L}^{\pm}_n] &= \pm 2\tilde{L}^{\pm}_{m+n},
	&
	[R_m, a^{\pm}_n] &= \pm a^{\pm}_{m+n}, 
	\notag \\[3pt]
	[ R_m, {\tilde{a}}^{\pm}_n] &= \pm {\tilde{a}}^{\pm}_{m+n}, 
	&
	[\tilde{R}_m,L^{\pm}_n] &= \pm 2\tilde{L}^{\pm}_{m+n}, 
	\notag \\[3pt]
	[\tilde{R}_m, \tilde{L}^{\pm}_n] &= \pm 2L^{\pm}_{m+n},
	&
	\{\tilde{R}_m, a^{\pm}_n \} &= \tilde{a}^{\pm}_{m+n}, 
	\notag \\[3pt]
	\{\tilde{R}_m, \tilde{a}^{\pm}_n\} &= a^{\pm}_{m+n}, 
	&
	[L^{+}_m,L^{-}_n] &= -R_{m+n} -6m\, \omega_{00}\, \delta_{m+n,0}, 
	\notag \\[3pt]
	[L^{\pm}_m, \tilde{L}^{\mp}_n] &= \mp \tilde{R}_{m+n}-6m\, \omega_{11} \delta_{m+n,0}, 
	&
	[\tilde{L}^{+}_m, \tilde{L}^{-}_n] &= -R_{m+n} - 6m\, \omega_{00}\, \delta_{m+n,0},  
	\notag \\[3pt]
	[L^{\pm}_m, \tilde{a}^{\mp}_n] &= \pm \tilde{a}^{\pm}_{m+n},& [L^{\pm}_m, a^{\mp}_n] &= \mp a^{\pm}_{m+n}, 
	\notag \\[3pt]
	\{ \tilde{L}^{\pm}_m, a^{\mp}_n\} &= -\tilde{a}^{\pm}_{m+n}, & \{\tilde{L}^{\pm}_m, \tilde{a}^{\mp}_n\} &= a^{\pm}_{m+n},
	\notag \\[3pt]
	\{a^{+}_m, a^{-}_n\} &= 2R_{m+n}+24m\omega_{00}\,\delta_{m+n,0}, & [a^{\pm}_m, \tilde{a}^{\mp}_n] &= \pm 2 \tilde{R}_{m+n}+24m\omega_{11}\delta_{m+n,0},
	\notag \\[3pt]
	\{ \tilde{a}^{+}_m, \tilde{a}^{-}_n \} &= 2R_{m+n}+24m\omega_{00}\,\delta_{m+n,0},  
	&
	[a^{\pm}_m, \tilde{a}^{\pm}_n] &= \mp 4 \tilde{L}^{\pm}_{m+n}, 
	\notag \\[3pt]
	\{a^{\pm}_m, a^{\pm}_n\} &= 4L^{\pm}_{m+n}, &	\{\tilde{a}^{\pm}_m, \tilde{a}^{\pm}_n\} &= -4L^{\pm}_{m+n}. \label{10Drels}
\end{alignat}
\begin{proof}
 Proof is similar to the case of $L(\g^8).$ 
We take $L' = \{  \  R_0 \ \}$ as the one-dimensional subalgebra. Then by Proposition \ref{P:SZP21},  the possible non-vanishing central elements are given by
\begin{equation}
	\begin{array}{cccc}
		\omega(R_m, R_n), & \quad \omega(\tilde{R}_m, \tilde{R}_n), & \quad \omega(L^+_m, L^-_n), & \quad \omega(\tilde{L}^+_m, \tilde{L}^-_n),
		\\
		\omega(a^+_m, a^-_n), & \quad \omega(\tilde{a}^+_m, \tilde{a}^-_n), & \quad \omega(R_m, \tilde{R}_n), & \omega(L^+_m, \tilde{L}^-_n),
		\\
		\omega(\tilde{L}^+_m, L^-_n), &   \omega(a^+_m, \tilde{a}^-_n), & \quad \omega(\tilde{a}^+_m, a^-_n).
	\end{array}	
\end{equation}

The loop $osp(1|2)$ subalgebras $ L(\mathfrak{h}_1) $ and $ L(\mathfrak{h}_2) $ (same definition as $L(\g^8)$) fix the following $(0,0)$-graded centers
\begin{align}
	\omega(R_m,R_n) &= 12m \delta_{m+n,0}\, \omega_{00}, \qquad \omega(L_m^+,L_n^-) = -6m\delta_{m+n,0}\, \omega_{00}, 
	\notag \\
	\omega(a^+_m,a^-_n) &=  \omega(\tilde{a}^+_m,\tilde{a}^-_n) = 24m \delta_{m+n,0}\, \omega_{00}.
	\label{10center1}
\end{align} 
The algebra $L(\g^{10})$ also has the loop $ sl(2) $ subalgebra spanned by $ R_m, \tilde{L}^{\pm}_m.$ 
It is known that $L(sl(2))$  has the unique central extension which should also exists in $L(\g^{10}).  $ The loop algebras $ L(sl(2))$,  $ L(\mathfrak{h}_1) $ and $ L(\mathfrak{h}_2) $ share $ R_m$, so the consistency requires  
\begin{equation}
	\omega(\tilde{L}_m^+,\tilde{L}_n^-) = -6m\delta_{m+n,0}\, \omega_{00}.
\end{equation}

The 2-cocycle condition for $ \tilde{R}_m, a^+_n, \tilde{a}^-_k$ yields
\begin{equation}
	2\omega(\tilde{R}_m, \tilde{R}_{n+k}) + \omega(a^+_n, a^-_{m+k}) + \omega(\tilde{a}^-_k, \tilde{a}^+_{m+n}) = 0. 
\end{equation}
Setting $ k = 0, $ we obtain that $ \omega(\tilde{R}_m,\tilde{R}_n) = 12m \delta_{m+n,0}\,\omega_{00}.$ 
 
We now turn to the $(1,1)$-graded centers, $ \omega(R_m,\tilde{R}_n), \omega(L^{\pm}_m, \tilde{L}^{\mp}_n) $ and $\omega(a^{\pm}_m, \tilde{a}^{\mp}_n). $ 
To determine these, we need to solve the following 2-cocycle conditions:

\begin{align}
	R_n, a_m^{\pm}, \tilde{a}_k^{\mp} \quad &\Rightarrow \quad 
	\omega(a_m^{\pm}, \tilde{a}^{\mp}_{n+k}) + 2 \omega(R_n, \tilde{R}_{k+m}) + \omega(\tilde{a}^{\mp}_k,a^{\pm}_{m+n}) = 0,
	\notag \\
	\tilde{R}_n, a^{\pm}_{m}, a^{\mp}_k \quad &\Rightarrow \quad 
	\omega(a^{\pm}_m, \tilde{a}^{\mp}_{n+k}) + 2\omega(\tilde{R}_n, R_{m+k}) + \omega(a^{\mp}_k, \tilde{a}^{\pm}_{m+n}) = 0,
	\notag \\
	L^{\mp}_n, a^{\pm}_m, \tilde{a}^{\pm}_k \quad &\Rightarrow \quad  
	\omega(a^{\pm}_m, \tilde{a}^{\mp}_{n+k}) - 4 \omega(L^{\mp}_n, \tilde{L}^{\pm}_{k+m}) + \omega(\tilde{a}^{\pm}_k, a^{\mp}_{m+n}) = 0,
	\notag \\
	\tilde{L}^{\mp}_n, a^{\pm}_m, a^{\pm}_k \quad &\Rightarrow \quad 
	\omega(a^{\pm}_m, \tilde{a}^{\mp}_{n+k}) -4\omega(\tilde{L}^{\mp}_n, L^{\pm}_{k+m}) + \omega(a^{\pm}_k, \tilde{a}^{\mp}_{m+n}) = 0,
	\notag \\
	R_m, L_n^+, \tilde{L}_k^- \quad &\Rightarrow \quad 
	\omega(R_m, \tilde{R}_{n+k}) - 2\omega(L_n^+, \tilde{L}^-_{k+m}) - 2\omega(\tilde{L}^-_k, L^+_{m+n}) = 0,
	\notag \\
	\tilde{R}_m, L^+_n, L^-_k \quad &\Rightarrow \quad 
	\omega(R_{n+k}, \tilde{R}_m) + 2 \omega(L_n^+, \tilde{L}^-_{k+m}) + 2\omega(L^-_k, \tilde{L}^+_{m+n}) = 0,
	\notag \\
	\tilde{R}_m, \tilde{L}^+_n, \tilde{L}^-_k \quad &\Rightarrow \quad 
	\omega(R_{n+k}, \tilde{R}_m) + 2\omega(\tilde{L}^+_n, L^-_{k+m}) + 2\omega(\tilde{L}^-_k, L^+_{m+n}) = 0,
	\notag \\
	R_m, L_n^-, \tilde{L}_k^+ \quad &\Rightarrow \quad 
	\omega(R_m, \tilde{R}_{n+k}) + 2\omega(\tilde{L}^+_{k+m}, L^-_n) + 2\omega(L^-_{m+n}, \tilde{L}^+_k) = 0.
\end{align}
It is not difficult to see the following solves all the conditions:
\begin{equation}
	\omega(X^a_m,X^b_n) = (-1)^{\bm{a}\cdot (1,1)} \eta^{ab} m \delta_{m+n,0}\,\omega_{11}.
\end{equation}
This completes the proof. 
\end{proof}

Let us now consider the $\Z2$-graded derivations. 
Like the invariant bilinear forms and the central extensions, we have a  derivation with non-trivial $\Z2$-degree.
\begin{proposition}
	\begin{enumerate}
		\item $d_{01}$ and $ d_{10}$ do not exist.
		\item $ [d_{00}, X_n^a] = n X_n^a $ is the derivation.
		\item $ d_{11} $ is defined by the followings:
		\begin{alignat}{2}
			[d_{11}, R_m] &= m \tilde{R}_m, & \qquad [d_{11}, L_m^{\pm}] &= m \tilde{L}_m^{\pm},
			\notag \\
			[d_{11}, \tilde{R}_m] &= m R_m, & [d_{11}, \tilde{L}_m^{\pm}] &= m L_m^{\pm},
			\notag \\
			\{d_{11}, a^{\pm}_m\} &= \pm m \tilde{a}^{\pm}_m, & \{ d_{11}, \tilde{a}^{\pm}_m\} &= \pm m a^{\pm}_m,
			\notag \\
			[d_{11}, \omega_{00}] &= 0, & [d_{11}, \omega_{11}] &= 0.
		\end{alignat}
	\end{enumerate}
\end{proposition}	
\begin{proof}
 The proof is the same as for Proposition \ref{Derivation8}, so we omit it. 	
\end{proof}

%
\section{$\Z2$-Graded Virasoro Algebras by Sugawara Construction} \label{SEC:Suga}

Thanks to the existence of non-degenerate invariant bilinear forms, we are able to discuss $\Z2$-graded extensions of the Virasoro algebras using $\widehat{\g}^8 $ and $ \widehat{\g}^{10}$ by Sugawara construction. 
In particular, we will have $(1,1)$-graded Virasoro current from the $(1,1)$-graded bilinear form.  

In this section, we will mainly follow  \cite{Schot} with appropriate $\Z2$-graded modifications. 
We introduce a current
\begin{equation}
	X^a(z) := \sum_{n \in \mathbb{Z}} X_n^a z^{-n-1}, \quad z \in \mathbb{C}
\end{equation}
and the formal delta function
\begin{equation}
	\delta(z-w) = \sum_{n \in \mathbb{Z}} z^{n-1} w^{-n} = \sum_{n \in \mathbb{Z}} z^n w^{-n-1} = \sum_{n \in \mathbb{Z}} z^{-n-1} w^n.
\end{equation}
The OPE of two currents is singular parts of the general Lie bracket:
\begin{align}
	\llbracket A(z),  B(w) \rrbracket &= \sum_{j=0}^{N-1} \frac{C^j(w)}{j!} \partial^j_w \delta(z-w)
	\notag \\
	&\ \Longleftrightarrow \ A(z) B(w) \sim \sum_{j=0}^{N-1} \frac{C^j(w)}{(z-w)^{j+1}}.
\end{align} 
The normal ordering of currents is defined by
\begin{equation}
	\NO{A(z) B(w)} = A(z)_+ B(w) + \ph{a}{b} B(w) A(z)_-
\end{equation}
where $ \bm{a} = \deg (A(z)), \bm{b} = \deg(B(w))$ and 
\begin{equation}
	A(z)_- := \sum_{n \geq 0} A_n z^{-n-1}, \quad
	A(z)_+ := \sum_{n < 0} A_n z^{-n-1}.
\end{equation}

\subsection{Witt algebra from $\widehat{\g}^8$}

In terms of currents, Eq.\eqref{8DrelExplicit} with the ordered basis Eq.\eqref{8Dorderedbasis} is rewritten as
\begin{equation}
	\llbracket X^a(z), X^b(w) \rrbracket = \sum_c \f{ab}{c}\, X^c(w) \delta(z-w) + g^{ab} \omega \partial_w \delta(z-w)
\end{equation} 
which is equivalent to the following OPE:
\begin{equation}
	X^a(z) X^b(w) \sim \frac{\sum_c\f{ab}{c}X^c(w)}{z-w} + \frac{g^{ab} \omega}{(z-w)^2}.
\end{equation}
This is the OPE of spin 1 current, so we define
	\begin{equation}
	L(z):= \frac{1}{2\omega+1} \sum_{a,b} g_{ab} \NO{X^a X^b}(z) \label{Su8D}
\end{equation}
and expect $L(z)$ is the Virasoro current. 
However, it turns out that this is not the case. 
One may see by the straightforward computation that
\begin{equation}
	\llbracket X^a(z), L(w) \rrbracket = X^a(w) \partial_w \delta(z-w)
\end{equation}
which is equivalent to
\begin{equation}
	X^a(z) L(w) \sim \frac{X^a(w)}{(z-w)^2}.
\end{equation}
Similarly, we have
\begin{equation}
   L(z) X^a(w) \sim \frac{X^a(w)}{(z-w)^2} + \frac{\partial_w X^a(w)}{z-w}.
\end{equation}
Using these one may see that
\begin{align}
	L(z) L(w) \sim \frac{2L(w)}{(z-w)^2} + \frac{\partial_w L(w)}{z-w} + \frac{\omega}{2\omega +1} \frac{\sum_{a,b} g_{ab} g^{ab}}{(z-w)^4}.
\end{align}
One may verify that the bilinear form satisfies
\begin{equation}
	\sum_{a,b} g_{ab} g^{ab} = 0.
\end{equation}
Thus, we have shown the following:
\begin{proposition}
	Eq.\eqref{Su8D} realizes the Witt algebra. 
\end{proposition}

\subsection{$\Z2$-graded Virasoro algebra form $ \widehat{\g}^{10}$}

Non-trivial extension of the Virasoro algebra will be obtained from $ \widehat{\g}^{10}$ because of the existence of $(1,1)$-graded central extension. 
In terms of currents, Eq.\eqref{10Drels} with the ordered basis Eq.\eqref{Ordered10D} is rewritten as
\begin{align}
	\llbracket X^a(z), X^b(w) \rrbracket &= \sum_c \f{ab}{c}\, X^c(w) \delta(z-w) 
	\notag \\
	&+(g^{ab} \, \omega_{00} 
	 + (-1)^{\bm{a}\cdot (1,1)} \eta^{ab} \omega_{11}) \partial_w \delta(z-w)
\end{align} 
which gives
\begin{align}
	X^a(z) X^b(w) &\sim \frac{\sum_c\f{ab}{c}X^c(w)}{z-w} 
	+ \frac{g^{ab} \, \omega_{00} + (-1)^{\bm{a}\cdot (1,1)} \eta^{ab} \omega_{11}}{(z-w)^2}.
\end{align}
Define
\begin{align}
	L^{00}_{(z)} &:= \frac{1}{2\omega_{00}+1} \sum g_{ab} \NO{X^a X^b}(z),
	\notag \\[3pt]
	L^{11}_{(z)} &:= \frac{1}{2\omega_{00}+1} \sum \eta_{ab} \NO{X^a X^b}(z),	
	\label{10DSugawara}
\end{align}
then after some computations we see that
\begin{align}
	X^a(z) L^{00}(w) &\sim \frac{X^a(w)}{(z-w)^2}  + \frac{\lambda_{11} \epsilon(X^a) \tilde{X}^a(w)}{(z-w)^2},
	\notag \\
	X^a(z) L^{11}(w) &\sim (-1)^{\bm{a}\cdot (1,1)} \left( \frac{\epsilon(X^a) \tilde{X}^a(w)}{(z-w)^2} + \frac{\lambda_{11} X^a(w)}{(z-w)^2} \right) \label{OPEXL}
\end{align}
where $ \tilde{X}^a $ denotes the current obtained from $ X^a(w) $ by shifting the $\Z2$-degree by $(1,1)$, for instance, $ R(w) \to \tilde{R}(w). $ 
We introduced the sign function $ \epsilon $ defined by
\begin{equation}
	\epsilon(X^a) = 
	\begin{cases}
		-1 & X^a = a^-,\; \tilde{a}^-
		\\
		+1 & \text{otherwise} 
	\end{cases}
\end{equation} 
and a $(1,1)$-graded constant
\begin{equation}
	\lambda_{11} = \frac{2\omega_{11}}{2\omega_{00}+1}.
\end{equation}

It may be helpful to give the OPE \eqref{OPEXL} in terms of the basis \eqref{10Dbasis}:
\begin{align}
	R(z) L^{00}(w) &\sim  \frac{R(w)}{(z-w)^2} + \lambda_{11} \frac{\tilde{R}(w)}{(z-w)^2},
	\\
	L^{\pm}(z) L^{00}(w) &\sim  \frac{L^{\pm}(w)}{(z-w)^2} + \lambda_{11} \frac{\tilde{L}^{\pm}(w)}{(z-w)^2},
	\\
	\tilde{R}(z) L^{00}(w) &\sim  \frac{\tilde{R}(w)}{(z-w)^2} + \lambda_{11} \frac{R(w)}{(z-w)^2},
	\\
	a^{\pm}(z) L^{00}(w) &\sim \frac{a^{\pm}(w)}{(z-w)^2} \pm \lambda_{11} \frac{\tilde{a}^{\pm}(w)}{(z-w)^2},
	\\
	\tilde{a}^{\pm}(z) L^{00}(w) &\sim \frac{\tilde{a}^{\pm}(w)}{(z-w)^2} \pm \lambda_{11} \frac{a^{\pm}(w)}{(z-w)^2}
\end{align}
and 
\begin{align}
	R(z) L^{11}(w) &\sim  \frac{\tilde{R}(w)}{(z-w)^2} + \lambda_{11} \frac{R(w)}{(z-w)^2},
	\\
	L^{\pm}(z) L^{11}(w) &\sim \frac{\tilde{L}^{\pm}(w)}{(z-w)^2} + \lambda_{11} \frac{L^{\pm}(w)}{(z-w)^2},
	\\
	\tilde{R}(z) L^{11}(w) &\sim  \frac{R(w)}{(z-w)^2} + \lambda_{11} \frac{\tilde{R}(w)}{(z-w)^2},
	\\
	a^{\pm}(z) L^{11}(w) &\sim \mp \frac{\tilde{a}^{\pm}(w)}{(z-w)^2} - \lambda_{11} \frac{a^{\pm}(w)}{(z-w)^2},
	\\
	\tilde{a}^{\pm}(z) L^{11}(w) &\sim \mp \frac{a^{\pm}(w)}{(z-w)^2} - \lambda_{11} \frac{\tilde{a}^{\pm}(w)}{(z-w)^2}.
\end{align}
After a lengthy computation using these relations,  we obtain the OPE for Virasoro currents
\begin{align}
	& L^{00}(z) L^{00}(w) \sim L^{11}(z) L^{11}(w) 
	\notag \\
	& \quad \sim  \frac{2(L^{00}(w) + \lambda_{11} L^{11}(w))}{(z-w)^2}
	+  \frac{\partial_w(L^{00}(w) + \lambda_{11}  L^{11}(w))}{z-w} 
	 + \frac{c_{00}/2}{(z-w)^4},
	\notag \\
	& L^{00}(z) L^{11}(w) \sim \frac{2(L^{11}(w) + \lambda_{11} L^{00}(w))}{(z-w)^2}
	+  \frac{\partial_w(L^{11}(w) + \lambda_{11}  L^{00}(w))}{z-w}
	\notag \\
	& \quad + \frac{c_{11}/2}{(z-w)^4}
\end{align}
where
\begin{equation}
	c_{00} := \frac{1}{3} \left( \frac{2\omega_{00}}{2\omega_{00} +1}  + \lambda_{11}^2 \right),
	\qquad
	c_{11} := \frac{\lambda_{11}}{3} \left( \frac{2\omega_{00}}{2\omega_{00} +1} +1 \right).
\end{equation}
These relations define a $\Z2$-graded generalization of the Virasoro algebra with two central extensions, one is not graded and the other is $(1,1)$-graded. 
The central elements are determined by those of $\widehat{\g}^{10}. $ 

 We introduce the modes $ L^{\vec{a}}_m$ from the Laurent expansion
\begin{equation}
	L^{\vec{a}}(z) = \sum_{n \in \mathbb{Z}} L^{\vec{a}}_m z^{-m-2}, \quad \vec{a} = 00, 11
\end{equation}
it is straightforward to derive the commutation relations of the $\Z2$-graded Virasoro algebra.  
In summary, we have established the following:
\begin{proposition}
Eq.\eqref{10DSugawara} realizes the infinite dimensional algebra defined by the relations
\begin{align}
	[L_m^{00}, L_n^{00}] &= [L_m^{11}, L_n^{11}]
	\notag \\
	&= (m-n) ( L_{m+n}^{00}  + \lambda_{11} L_{m+n}^{11}) + \frac{c_{00}}{12}  m (m^2-1) \delta_{m+n,0},
	\notag \\[3pt]
	[L_m^{00}, L_n^{11}] 
	&= (m-n) (L_{m+n}^{11}  +  \lambda_{11} L_{m+n}^{00}) + \frac{c_{11}}{12}  m (m^2-1) \delta_{m+n,0}.	
\end{align}
\end{proposition}

%
\section{Concluding Remarks} \label{SEC:CR}

We introduced the affine extensions of the $\Z2$-graded $osp(1|2)$ Lie superalgebras $\g^8 $ and $ \g^{10}.$ 
The main problem of the extension was the central extensions of the loop algebras $ L(\g^8)$ and $ L(\g^{10})$. 
This was solved by using a cohomological argument with the help of the observation that the algebras $\g^8$ and $ \g^{10}$ have non-degenerate invariant bilinear forms which is also a main result of the present work. 
The algebra $\g^{10} $ has $(1,1)$-graded bilinear form, however, none of the algebras have  $(1,0)$ or $(0,1)$-graded forms. To reflect this property, $\widehat{\g}^{10}$ has a $(1,1)$ graded central element and a $(1,1)$-graded derivation. 

The affine $\Z2$-graded algebras $\widehat{\g}^8 $ and $ \widehat{\g}^{10}$ were used by the Sugawara construction to discuss possible $\Z2$-graded extensions of the Virasoro algebra. 
A non-trivial extension with two central elements was obtained from $\widehat{\g}^{10}.$ 
The $\Z2$-graded Virasoro algebra consists of $(0,0)$ and $(1,1)$-graded currents, but it has no  fermionic currents due to the non-existence of $(0,1)$ and $(1,0)$ invariant bilinear forms. 
We have done a brute force search for the missing currents, however, it is still an open problem. 
We mention that other kinds of $\Z2$-extensions of the Virasoro algebra have been discussed in \cite{zhe,NAPSIJSinf}

Besides integrable systems related to $\widehat{\g}^8,$ or $ \widehat{\g}^{10}$, which were our motivation for the present work, it is interesting and important to study representations of these algebras. As shown in \cite{NAKA}, the irreps of the finite dimensional $\g^{10}$ are much richer than those of the superalgebra $osp(1|2)$. 
It is expected that the irreps of $\widehat{\g}^8 $ and $\widehat{\g}^{10}$ are also richer than those of $\widehat{osp(1|2)}.$ 

In the present work, we have considered the simplest $\Z2$-graded $osp(1|2)$ Lie superalgebra.  
The results presented here will be generalized to the larger class of $\Z2$-graded Lie superalgebras. In particular, the generalization to the $\Z2$-graded $osp$ series will be possible. Since its defining relations are explicitly given  in the literature, one can work out all the details explicitly \cite{GrJa,StoVDJ5}.

\section*{Acknowledgments}

N. A. is grateful to Francesco Toppan and Zhanna Kuznetsova for valuable discussions. 
He is supported by JSPS KAKENHI Grant Number JP23K03217. J. S. would like to thank N. Aizawa for his warm hospitality and financial support at the Osaka Metropolitan University where most part of this work was done

\end{document}